\documentclass[10pt,conference,letterpaper]{IEEEtran}

\usepackage{amsmath,amssymb,amsthm,bm}
\usepackage{booktabs}
\usepackage{graphicx}
\usepackage{array}
\usepackage{tabularx}
\usepackage{algorithm}
\usepackage{algorithmic}
\usepackage[hidelinks]{hyperref}

\newtheorem{theorem}{Theorem}
\newtheorem{lemma}{Lemma}
\newtheorem{proposition}{Proposition}
\newcommand{\T}{\mathcal{T}}
\newcommand{\M}{\mathcal{M}}
\newcommand{\F}{F}

\begin{document}

\title{PrefixPlace: Provable Prefix Key--Value Placement for Large Language Model Serving under Heterogeneous Compute and Transfer Costs}
\author{
  \IEEEauthorblockN{Zhiyu Wang and Rajkumar Buyya}
  \IEEEauthorblockA{\textit{Quantum Cloud Computing and Distributed Systems (qCLOUDS) Lab} \\
  \textit{School of Computing and Information Systems} \\
  \textit{The University of Melbourne, Australia} \\
  \{zhiyu.wang4, rbuyya\}@unimelb.edu.au}
}
\maketitle

\begin{abstract} 
Prefix Key--Value (KV) reuse avoids repeated prefill in Large Language Model (LLM) inference, but local misses require recomputation or replica fetches. Their relative cost varies with hardware, prefix depth, KV goodput, and replica location, making hit-rate-based placement suboptimal. To address this issue, we propose an epoch-level planner, PrefixPlace, which assigns prefix-complete targets under memory budgets and profiled demand, compute, and transfer costs. The objective decomposes into local-copy value plus first-replica coverage, and source-dependent costs yield a monotone facility-location objective; each worker update is an additive rooted-tree problem solved exactly in $O(nk)$ time for $n$ chunks and capacity $k$, giving a fixed-order $1/2$-approximation that coordinate refinement and order-diverse starts improve without weakening. T4, L4, and A100 measurements reveal distinct regimes. Across 432 instances with exact optima, PrefixPlace averages 99.84\% of optimum and never falls below 98.02\%. In Retrieval-Augmented Generation (RAG) replays, it improves materialization-cost saving by 40.3\% over vLLM Automatic Prefix Caching (vLLM-APC) and 6.3\% over the best offline baseline. On WikiQA, gains are 40.4\% and 5.3\%. Finally, PrefixPlace solves a 50,000-node, 16-worker placement in 12.3 s on one processor, enabling timely replanning. 
\end{abstract}


\section{Introduction}

Prefix reuse avoids repeated prefill when Large Language Model (LLM) requests share system prompts, retrieved documents, few-shot templates, or conversation histories. Modern inference engines reuse the corresponding Key--Value (KV) attention state~\cite{vllm,sglang,promptcache}, and recent designs also move KV state across workers or storage tiers~\cite{cachedattention,cachegen,mooncake,impress2025}. These mechanisms provide reuse and transfer primitives, but leave a separate planning question: \emph{which reusable prefixes should remain resident at each worker?}

This placement decision cannot be reduced to cache hit rate. On a local miss, a worker may either recompute the missing KV state from the token transcript or fetch a compatible replica. Their relative cost changes with accelerator speed, prefix depth, KV footprint, and effective KV goodput. Figure~\ref{fig:hw} isolates this effect using the same Qwen2.5-3B-Instruct model, 16-bit Floating-Point (FP16) representation, prompts, and software stack on three Graphics Processing Units (GPUs). Under a common effective KV goodput of 1.5 Gb/s, fetching the measured 18.9-MB KV state for a 512-token chunk takes about 100.7 ms. Recomputation is already slower on T4, crosses the transfer cost near 3.8K existing tokens on L4, and remains faster throughout the measured A100 range. Thus, the same model and transfer condition can require opposite decisions solely because requester hardware changes.

\begin{figure}[t]
\centering
\includegraphics[width=.80\linewidth]{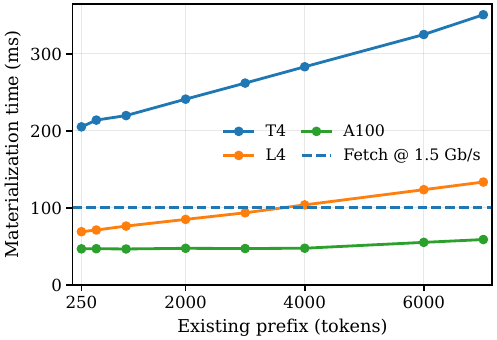}
\caption{Measured time to materialize a 512-token missing chunk of Qwen2.5-3B in 16-bit Floating-Point (FP16) after an existing prefix. The dashed line is the fetch time implied by the measured 18.9-MB KV footprint and 1.5-Gb/s effective KV goodput. T4 favors fetching, L4 crosses near 3.8K tokens, and A100 favors recomputation over the measured range.}
\label{fig:hw}
\end{figure}

To address this problem, we develop \emph{PrefixPlace}, an epoch-level planner that assigns each worker a resident prefix target from an observed prefix tree, demand matrix, memory budget, and profiled materialization costs. A feasible target is \emph{prefix-complete}: retaining a chunk also retains every cacheable ancestor on its prefix path.

Coordination is needed because one resident replica can reduce miss cost for multiple requesters, but only where fetching beats recomputation; when transfer cost depends on the source, replica identity also affects value. Local-popularity policies can therefore over-replicate shared prefixes, whereas deduplication alone ignores requester heterogeneity and source-dependent transfer. PrefixPlace accepts arbitrary epoch-level compute and transfer profiles rather than a parametric latency law; every local copy then contributes modular value, the first replica contributes weighted coverage, and source-dependent costs generalize this to a monotone facility-location objective. In both cases, one worker's exact marginal value is additive across tree nodes, enabling a shared rooted-tree oracle with provable guarantees.

Our contributions are threefold.
\begin{itemize}
\item \textbf{Cost-aware prefix placement model.} We formulate epoch-level placement of prefix-complete KV targets under arbitrary profiled compute and transfer costs. We derive an exact modular-plus-coverage decomposition for requester-side transfer costs and a source-dependent facility-location formulation, precisely characterizing the shared value created by the first replica of each chunk.
\item \textbf{Algorithms with guarantees.} For a prefix tree with $n$ placement chunks and a worker capacity of $k$ chunks, we give an $O(nk)$ exact single-worker Dynamic Programming (DP) algorithm. We prove that coordinated placement is strongly NP-hard and develop WorkerGreedy, a fixed-order $1/2$-approximation for both objectives. PrefixPlace combines this guaranteed initialization with coordinate refinement and order-diverse starts to improve solution quality while preserving the bound.
\item \textbf{Profile-driven evaluation.} Measurements on T4, L4, and A100 GPUs across Qwen2.5 models from 3B to 32B expose hardware- and goodput-dependent fetch/recompute regimes. PrefixPlace averages 99.84\% of exact optimum over 432 requester-side and 99.62\% over 45 source-dependent instances, and Retrieval-Augmented Generation (RAG) replay, WikiQA exact-demand, topology, demand-shift, and Central Processing Unit (CPU) scaling studies evaluate practical advantage and replanning cost.
\end{itemize}

The rest of the paper is organized as follows. Sections~\ref{sec:measurement}--\ref{sec:eval} present the problem setting, optimization model, algorithms, and performance evaluation respectively; Section~\ref{sec:related} reviews related work and Section~\ref{sec:conclusion} concludes with future directions.

\section{Problem Setting}
\label{sec:measurement}

\subsection{Prefix KV as a placement object}
A transformer prefix determines the KV tensors reusable by every continuation of the same token sequence. Modern engines expose this reuse through block-structured prefix caches~\cite{vllm,sglang}, while recent designs stream, offload, or share KV state~\cite{cachegen,mooncake,cachedattention,droidspeak}. PrefixPlace groups contiguous engine blocks into \emph{placement chunks}, each treated as one placement and materialization unit. Grouping preserves exact-prefix semantics, amortizes fixed operation overhead, and limits optimizer state.

For each worker, the planner selects a prefix-complete resident target: a parent-closed set in which selecting a chunk also selects every cacheable ancestor. A target may contain multiple branches, but each selected chunk lies on a complete locally resident prefix path. A request's materialization cost is the sum of the incremental costs of its required chunks; PrefixPlace consumes the profiled lookup values directly and does not assume a linear depth-cost relationship.

\subsection{Compute and transfer cost model}
Let $w_m(b)$ be the profiled cost for worker $m$ to materialize placement chunk $b$ locally, conditioned on its preceding prefix, and let $f_m(b)$ be the effective cost to fetch that chunk from an eligible peer. For a chunk of $Z_b$ bytes and effective KV goodput $r_m$,
\begin{equation}
 f_m(b)=Z_b/r_m.
\label{eq:fetchcost}
\end{equation}
The planner may instead consume a directly measured fetch-cost table, including stable copy, protocol, and round-trip components. When holder identity matters, $f_{m\leftarrow h}(b)$ denotes the cost for requester $m$ to fetch $b$ from holder $h$.

Remote reuse is beneficial to requester $m$ exactly when $w_m(b)>f_m(b)$. We call a measured depth where the preferred action changes a \emph{fetch/recompute crossover}; the optimization neither requires one to exist nor assumes monotone costs, and hardware, model scale, prefix depth, and network path can each flip the preferred action for the same logical chunk.

\subsection{Planning epoch and compatibility}
PrefixPlace optimizes one planning epoch from aggregate prefix-tree demand, worker budgets, compatible resident states, and effective compute/transfer profiles. Routing operates outside the optimizer and determines the epoch demand matrix; Section~\ref{sec:alg} describes when refreshed inputs trigger replanning. Candidate holders are filtered for compatible model weights and revision, KV precision, positional encoding, and parallelism layout. The planner exports one prefix-complete target per worker to the exact-prefix cache layer.

\section{Optimization Model and Structural Decomposition}
\label{sec:model}

\subsection{Prefix tree, demand, and feasibility}
Let $\T=(V,E)$ be a rooted prefix tree with $n=|V|$ equal-sized cacheable placement chunks; an optional no-cost dummy root represents the empty prefix and does not consume budget. Each node extends the exact prefix represented by its parent. Worker $m\in\M=\{1,\ldots,M\}$ has budget $k_m$ and selects a resident set $T_m\subseteq V$. We call $T_m$ \emph{prefix-complete} when it is parent-closed: selecting a node also selects every cacheable ancestor. A target may contain multiple branches, but every selected node lies on a complete locally resident prefix path. Workers optimized together satisfy the compatibility filters of Section~\ref{sec:measurement}, so $k_m$ maps directly to a byte budget.

Let $\Lambda_m(b)\ge0$ be the aggregate epoch demand for requests routed to worker $m$ that require chunk $b$. A request contributes demand to every placement chunk on its prefix path. When $\Lambda$ is measured in requests per epoch and $w,f$ in milliseconds, the objective below is milliseconds saved per epoch; normalized demand weights induce the same optimizer. The local recomputation cost is $w_m(b)\ge0$. We first consider a source-independent, requester-specific fetch cost $f_m(b)\ge0$, which may vary by requester and chunk; Section~\ref{sec:sourceaware} then admits arbitrary source-dependent costs.

Once resident locally, a chunk incurs zero materialization cost for that requester. If only a peer stores it, the requester chooses the cheaper of recomputation and fetch, paying $\min\{w_m(b),f_m(b)\}$; with no resident copy, it pays $w_m(b)$. Thus, a remote replica creates value only for requesters whose transfer cost is below recomputation.

\subsection{Savings decomposition}
Use an all-recompute placement as the baseline. Define the local value
\begin{equation}
 a_m(b)=\Lambda_m(b)\min\{w_m(b),f_m(b)\},
\label{eq:am}
\end{equation}
and the aggregate remote-coverage value
\begin{equation}
 \Delta(b)=\sum_{m=1}^{M}\Lambda_m(b)\,[w_m(b)-f_m(b)]_+,
\label{eq:delta}
\end{equation}
where $[x]_+=\max\{x,0\}$. The local term measures what worker $m$ gains from storing $b$ even when another copy exists. The coverage term measures the shared saving created by the \emph{first} resident copy. This separation is exact, not an approximation.

\begin{lemma}[Modular-plus-coverage decomposition]
\label{lem:decomp}
For any feasible placement $\bm T=(T_1,\ldots,T_M)$, total materialization-cost savings are exactly
\begin{equation}
\F(\bm T)=
\sum_{m=1}^{M}\sum_{b\in T_m} a_m(b)
+\sum_{b\in \cup_m T_m}\Delta(b).
\label{eq:objective}
\end{equation}
Consequently $\F:2^{\mathcal E}\!\to\mathbb R_+$ is nonnegative, monotone, and submodular over the ground set $\mathcal E=\{(m,b):m\in\M,\ b\in V\}$ of worker--chunk placement elements.
\end{lemma}
\begin{proof}
Fix requester $m$ and chunk $b$. If no worker stores $b$, the saving is zero. If a peer but not $m$ stores it, the saving is $[w_m(b)-f_m(b)]_+$. If $m$ stores it, the saving is $w_m(b)=\min\{w_m(b),f_m(b)\}+[w_m(b)-f_m(b)]_+$. Multiplying by $\Lambda_m(b)$, summing over requesters, and observing that the second summand is earned once iff $b$ is covered gives~\eqref{eq:objective}. The first term is modular and the second is weighted coverage, hence monotone submodular.
\end{proof}

The decomposition reduces arbitrary compute and transfer profiles to two interpretable values. If every requester prefers recomputation, $\Delta(b)=0$ and coordination creates no remote-reuse value for chunk $b$. Otherwise, the first replica contributes exactly $\Delta(b)$ beyond the holder's local benefit. Hence $\Delta$ identifies where coordination can matter, while budgets and parent closure determine which opportunities are feasible.

\subsection{Source-dependent transfer costs}
\label{sec:sourceaware}
Requester-side costs treat eligible holders as equivalent, while transfer costs may also depend on the source. Let $f_{m\leftarrow h}(b)$ be the cost for requester $m$ to fetch chunk $b$ from holder $h$, and define the saving offered by holder $h$ as
\begin{equation}
 q_{m\leftarrow h}(b)=\begin{cases}
 w_m(b), & h=m,\\
 [w_m(b)-f_{m\leftarrow h}(b)]_+, & h\ne m.
 \end{cases}
\label{eq:qpair}
\end{equation}
If $H_b(\bm T)=\{h:b\in T_h\}$ is the holder set induced by placement $\bm T$, then the exact saving is
\begin{equation}
 \F_{\rm pair}(\bm T)=\sum_{b\in V}\sum_{m=1}^{M}\Lambda_m(b)
 \max_{h\in H_b(\bm T)}q_{m\leftarrow h}(b),
\label{eq:pair}
\end{equation}
where the maximum over an empty set is zero.

\begin{proposition}[Source-dependent marginal oracle]
\label{prop:source}
$\F_{\rm pair}$ is nonnegative, monotone, and submodular. Moreover, after fixing the current placements of all workers other than $h$, its exact marginal placement problem is a parent-closed tree selection with additive node profits
\begin{equation}
 p_h(b)=\sum_m\Lambda_m(b)[q_{m\leftarrow h}(b)-v_m^{(-h)}(b)]_+,
\label{eq:pairmarginal}
\end{equation}
where $v_m^{(-h)}(b)$ is requester $m$'s best saving on $b$ among the current holders other than $h$.
\end{proposition}
\begin{proof}
For each $(m,b)$, the term in~\eqref{eq:pair} is the maximum of fixed nonnegative weights over selected holders, which is monotone submodular. Their nonnegative weighted sum is therefore monotone submodular. Adding holder $h$ improves $(m,b)$ by exactly $[q_{m\leftarrow h}(b)-v_m^{(-h)}(b)]_+$; summing over requesters gives~\eqref{eq:pairmarginal}, and summing these independent chunk marginals over a candidate subtree gives its exact gain.
\end{proof}
Source dependence changes the marginal node weights but not the per-worker tree problem. The rooted-tree oracle and the approximation/refinement framework therefore apply to both objectives. Setting $f_{m\leftarrow h}(b)=f_m(b)$ for every $h\ne m$ recovers~\eqref{eq:objective}.

The placement problem is therefore
\begin{align}
 \max_{T_1,\ldots,T_M}\quad & \F(\bm T) \label{eq:problem}\\
 \text{s.t.}\quad & |T_m|\le k_m,\quad T_m\text{ parent-closed},\ \forall m,\nonumber
\end{align}
with $\F_{\rm pair}$ replacing $\F$ for source-dependent costs. Compute and transfer behavior enter through profiled node costs; no ordering, convexity, or parametric latency law is required. Lookup tables, interpolation, and analytical profiles therefore lead to the same combinatorial problem and the same guarantees.

\section{Algorithms and Planner Operation}
\label{sec:alg}

Figure~\ref{fig:arch} summarizes how PrefixPlace turns epoch inputs into per-worker targets. This section develops each stage in turn. 

\begin{figure*}[t]
\centering
\includegraphics[width=\textwidth]{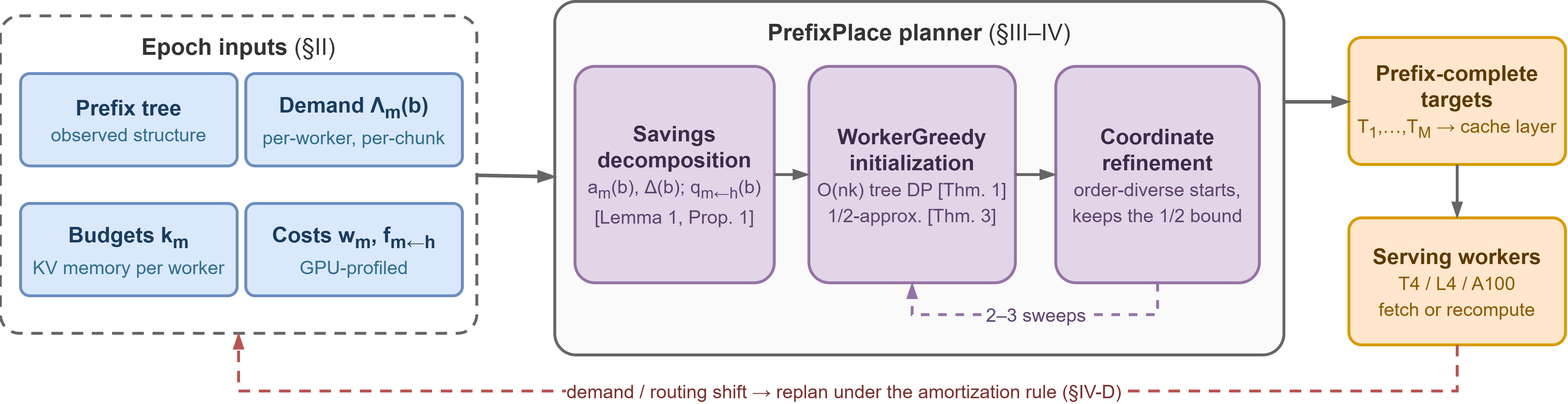}
\caption{PrefixPlace architecture. Each epoch, the observed prefix tree, per-worker demand matrix $\Lambda_m(b)$, memory budgets $k_m$, and GPU-profiled compute/transfer costs $w_m(b)$, $f_{m\leftarrow h}(b)$ feed a three-stage planner: exact savings decomposition into local value, coverage, and source-aware marginals (Lemma~\ref{lem:decomp}, Prop.~\ref{prop:source}); WorkerGreedy initialization with the exact $O(nk)$ rooted-tree oracle and a fixed-order $1/2$-approximation (Theorems~\ref{thm:dp} and~\ref{thm:half}); and monotone coordinate refinement over order-diverse starts that preserves the $1/2$ bound. One prefix-complete target per worker is exported to the exact-prefix cache layer; demand or routing shifts trigger replanning under the amortization rule (\S\ref{sec:alg}-D).}
\label{fig:arch}
\end{figure*}

\subsection{Exact single-worker oracle}
The shared algorithmic primitive is a maximum-profit parent-closed subtree for one worker. Let each node have profit $p(b)$ and let the budget be $k$. Number nodes in preorder $b_1,\ldots,b_n$, and let $\mathrm{next}(i)$ be the first position after the subtree rooted at $b_i$. Define $D[i,j]$ as the maximum profit obtainable from positions $i,\ldots,n$ using at most $j$ additional chunks, under the invariant that the nodes selected before position $i$ are parent-closed and contain every ancestor of $b_i$. For $i\le n$ and $j\ge1$,
\begin{equation}
D[i,j]=\max\!\left\{
D[\mathrm{next}(i),j],\;
 p(b_i)+D[i+1,j-1]
\right\}.
\label{eq:dp}
\end{equation}
The boundary conditions are $D[n+1,j]=0$ for all $j$ and $D[i,0]=0$ for all $i$. The first branch skips $b_i$ and therefore its entire subtree; the second includes $b_i$, so preorder position $i+1$ is eligible. Subtree endpoints, and hence every $\mathrm{next}(i)$, are precomputed in $O(n)$ time.

\begin{theorem}[Exact rooted-tree oracle]
\label{thm:dp}
For a rooted tree with $n$ placement chunks and budget $k$, the maximum-profit parent-closed placement is computed exactly in $O(nk)$ time and $O(nk)$ memory.
\end{theorem}
\begin{proof} 
In preorder, the subtree rooted at $b_i$ occupies the contiguous interval $\{b_i,\ldots,b_{\mathrm{next}(i)-1}\}$. Consider an eligible state $(i,j)$ in which the previously selected nodes are parent-closed and contain every ancestor of $b_i$. Every feasible continuation makes one of two decisions. If it omits $b_i$, parent closure also excludes every descendant of $b_i$, so the next state is $(\mathrm{next}(i),j)$. If $\mathrm{next}(i)\le n$, every ancestor of $b_{\mathrm{next}(i)}$ is also an ancestor of $b_i$ and is therefore already selected. If the continuation includes $b_i$, it earns $p(b_i)$ and proceeds to $(i+1,j-1)$. When $b_{i+1}$ is a descendant of $b_i$, all of its ancestors are either ancestors of $b_i$ or $b_i$ itself and are therefore selected; when $b_i$ is a leaf, $i+1=\mathrm{next}(i)$ and the preceding argument applies. Thus, both successor states preserve eligibility, and the two cases are exhaustive, establishing~\eqref{eq:dp} by backward induction. Hence $D[1,k]$ is optimal. A preorder traversal computes all subtree endpoints in $O(n)$ time. The table contains $O(nk)$ states with $O(1)$ work per state, and backtracking visits at most $n$ states. 
\end{proof}

The resulting dynamic program avoids generic child-by-child tree-knapsack convolution and is linear in tree size for each budget unit.

\subsection{Joint placement: hardness and approximation}
Coordinated placement couples workers through shared coverage value. The resulting hardness is intrinsic rather than an artifact of heterogeneous hardware or transfer costs.

\begin{theorem}[Strong NP-hardness]
\label{thm:hard}
Maximizing~\eqref{eq:objective} over parent-closed placements with per-worker budgets is strongly NP-hard, even when all workers share the same depth-dependent recomputation and fetch costs and all holder workers have the same positive budget.
\end{theorem}
\begin{proof}
Reduce from 3-Partition. Given integers $s_1,\ldots,s_{3q}$ with
$\sum_i s_i=qB$ and $B/4<s_i<B/2$, construct a spider with an
always-present dummy root and one leg of $s_i$ cacheable chunks for each item.
Create $q$ holder workers of budget $B$ and one requester worker of budget zero.
All workers have the same costs: a depth-$d$ chunk has recomputation cost
$w(d)=2d$ and fetch cost $f=1$. Only the requester has demand. Put one request
at the tip of leg $i$ with rate $\lambda_i=1/s_i$; hence every chunk on that
leg has demand $\lambda_i$ at the requester.

If the union of holder placements covers the first $x_i\le s_i$ chunks of leg
$i$, its saving is
\[
 \lambda_i\sum_{d=1}^{x_i}(w(d)-f)
 =\frac{1}{s_i}\sum_{d=1}^{x_i}(2d-1)
 =\frac{x_i^2}{s_i}\le x_i,
\]
with equality only for $x_i\in\{0,s_i\}$. Therefore total saving is at most
$\sum_i x_i\le qB$, the aggregate holder capacity. Reaching $qB$ requires every
covered leg to be complete, every item leg to be covered, no duplicated chunk,
and every holder budget to be full. Because a placement is parent-closed,
covering the tip of a leg forces one holder to store that entire leg. Thus
value $qB$ is attainable iff the item lengths can be assigned to the $q$
holders with load exactly $B$ each. The 3-Partition bounds then force exactly
three items per holder. Conversely, any valid 3-partition gives such a
placement and attains $qB$. The construction uses only polynomially encoded
integers/rationals and preserves the strong NP-hardness of 3-Partition~\cite{garey-johnson}.
\end{proof}

The decomposition nevertheless preserves exact per-worker optimization. For a fixed worker order, let $U$ be the chunks already covered by earlier workers. Worker $m$ then sees the exact marginal profit
\begin{equation}
 p_m^U(b)=a_m(b)+\mathbf{1}\{b\notin U\}\Delta(b).
\label{eq:marginal}
\end{equation}
The best feasible placement for $m$ is therefore exactly Theorem~\ref{thm:dp} with profits $p_m^U$.

\begin{theorem}[WorkerGreedy]
\label{thm:half}
Sequentially optimizing each worker's exact marginal contribution in any order fixed before the run produces a feasible placement $\bm T^{G}$ with
$G(\bm T^{G})\ge\tfrac12G(\bm T^*)$, where $G$ is either $\F$ in~\eqref{eq:objective} or the source-dependent $\F_{\rm pair}$ in~\eqref{eq:pair}.
\end{theorem}
\begin{proof}
View each placement element as a worker--chunk pair. Let $S_i$ be the union of the first $i$ greedy worker placements in the fixed order, and let $O_i$ denote the placement assigned to worker $i$ in an optimal joint solution $\bm T^*$. Write $R_i=S_M\cup O_1\cup\cdots\cup O_i$, with $R_0=S_M$. Monotonicity gives $G(\bm T^*)\le G(R_M)$. Moreover, $S_{i-1}\subseteq R_{i-1}$, so diminishing returns yields
\[
G(R_i)-G(R_{i-1})
\le G(S_{i-1}\cup O_i)-G(S_{i-1}).
\]
Summing over $i$ gives
\begin{align*}
G(\bm T^*)
&\le G(S_M)+\sum_i\bigl[G(S_{i-1}\cup O_i)-G(S_{i-1})\bigr].
\end{align*}
When worker $i$ is processed, $O_i$ is feasible for its exact marginal oracle. The corresponding term is therefore at most $G(S_i)-G(S_{i-1})$. These greedy marginals telescope to $G(S_M)$, and hence $G(\bm T^*)\le2G(S_M)=2G(\bm T^G)$.
\end{proof}

Configuration methods based on Linear Programming (LP) give stronger worst-case factors for the requester-side separable special case~\cite{fgms06}. PrefixPlace contributes a direct rooted-tree-oracle framework that applies unchanged to both requester-side coverage and source-dependent facility-location marginals, while avoiding a global configuration LP and rounding stage.

\subsection{PrefixPlace: guaranteed coordination and refinement}
PrefixPlace strengthens the guaranteed WorkerGreedy solution through coordinate refinement: it removes one worker at a time, recomputes that worker's exact node marginals against all other placements, reruns the same DP using~\eqref{eq:marginal} or~\eqref{eq:pairmarginal}, and accepts only strictly improving replacements (tolerance $\varepsilon$ in floating point). Because the feasible state space is finite, refinement terminates at a coordinate-wise local optimum; monotone improvement preserves the starting $1/2$ guarantee.

To reduce order sensitivity, PrefixPlace refines multiple WorkerGreedy starts and returns the best candidate. For $M\le5$, it evaluates every order; otherwise, it uses the standalone-score order, its reverse, and six fixed pseudorandom permutations. For the requester-side objective, it also includes Independent, Popularity, and LocalDedup candidates, ensuring that the returned placement is never worse than these alternatives under the same objective.

Given node marginals, one requester-side greedy pass or one refinement sweep costs $O\!\left(n\sum_m k_m\right)$. For source-dependent costs, maintaining the largest and second-largest current holder values for each requester--chunk pair computes all leave-one-worker-out marginals in $O(M^2n)$ per sweep. These are per-pass bounds; the evaluated instances converge in two to three sweeps.

\begin{algorithm}[t]
\caption{\textsc{PrefixPlace} with the exact rooted-tree oracle.}
\label{alg:prefixplace}
\scriptsize
\setlength{\algorithmicindent}{0.92em}
\begin{algorithmic}[1]
\REQUIRE Tree $\T$; budgets $\{k_m\}$; demand $\{\Lambda_m\}$; profiled costs; tolerance $\varepsilon$
\ENSURE Prefix-complete placements $\bm T=(T_1,\ldots,T_M)$
\STATE \textbf{procedure} $\textsc{RootedDP}(p,k)$
\STATE Number nodes $b_1,\ldots,b_n$ in preorder and compute $\mathrm{next}(i)$
\STATE Set $D[n+1,j]\leftarrow0$ for $0\le j\le k$
\STATE Set $D[i,0]\leftarrow0$ for $1\le i\le n$
\FOR{$i=n,n-1,\ldots,1$}
  \FOR{$j=1,2,\ldots,k$}
    \STATE $s\leftarrow D[\mathrm{next}(i),j]$ \COMMENT{skip $b_i$ and its subtree}
    \STATE $t\leftarrow p(b_i)+D[i+1,j-1]$ \COMMENT{take $b_i$}
    \STATE $D[i,j]\leftarrow\max\{s,t\}$ and record the maximizing branch
  \ENDFOR
\ENDFOR
\STATE Backtrack from $D[1,k]$ to recover the parent-closed set
\STATE \textbf{return} the recovered set
\STATE \textbf{procedure} $\textsc{PrefixPlace}(\T,\{k_m\},\{\Lambda_m\},\textit{costs})$
\STATE Build $a_m,\Delta$ by~\eqref{eq:am}--\eqref{eq:delta}, or $q_{m\leftarrow h}$ by~\eqref{eq:qpair}
\STATE $\Pi\leftarrow$ all worker orders if $M\le5$
\IF{$M>5$}
  \STATE $\Pi\leftarrow$ standalone-score order, its reverse, and six fixed pseudorandom permutations
\ENDIF
\STATE $\mathcal C\leftarrow\emptyset$
\FOR{each fixed order $\pi\in\Pi$}
  \STATE $T_m\leftarrow\emptyset$ for every worker $m$
  \FOR{each worker $h$ in order $\pi$}
    \IF{requester-side costs}
      \STATE $U\leftarrow\bigcup_{\ell\ne h}T_\ell$
      \STATE $p_h(b)\leftarrow a_h(b)+\mathbf1\{b\notin U\}\Delta(b)$ for all $b$
    \ELSE
      \STATE $v_m^{(-h)}(b)\leftarrow\max_{\ell\ne h:b\in T_\ell}q_{m\leftarrow\ell}(b)$
      \STATE $p_h(b)\leftarrow\sum_m\Lambda_m(b)[q_{m\leftarrow h}(b)-v_m^{(-h)}(b)]_+$
    \ENDIF
    \STATE $T_h\leftarrow\textsc{RootedDP}(p_h,k_h)$
  \ENDFOR
  \REPEAT
    \STATE $\textit{improved}\leftarrow\textsc{false}$
    \FOR{each worker $h$ in order $\pi$}
      \STATE $T_h^{\rm old}\leftarrow T_h$; $g_{\rm old}\leftarrow G(\bm T)$; $T_h\leftarrow\emptyset$
      \IF{requester-side costs}
        \STATE $U\leftarrow\bigcup_{\ell\ne h}T_\ell$
        \STATE $p_h(b)\leftarrow a_h(b)+\mathbf1\{b\notin U\}\Delta(b)$ for all $b$
      \ELSE
        \STATE Recompute $v_m^{(-h)}(b)$ from the current holders
        \STATE $p_h(b)\leftarrow\sum_m\Lambda_m(b)[q_{m\leftarrow h}(b)-v_m^{(-h)}(b)]_+$
      \ENDIF
      \STATE $T'_h\leftarrow\textsc{RootedDP}(p_h,k_h)$
      \IF{$G(\bm T_{-h},T'_h)>g_{\rm old}+\varepsilon$}
        \STATE $T_h\leftarrow T'_h$; $\textit{improved}\leftarrow\textsc{true}$
      \ELSE
        \STATE $T_h\leftarrow T_h^{\rm old}$
      \ENDIF
    \ENDFOR
  \UNTIL{$\textit{improved}=\textsc{false}$}
  \STATE $\mathcal C\leftarrow\mathcal C\cup\{\mathrm{copy}(\bm T)\}$
\ENDFOR
\STATE Add Independent, Popularity, and LocalDedup placements to $\mathcal C$ when requester-side
\STATE \textbf{return} $\arg\max_{\bm X\in\mathcal C}G(\bm X)$
\end{algorithmic}
\end{algorithm}

\subsection{Planner operation and epoch updates}
Algorithm~\ref{alg:prefixplace} turns each epoch's tree, demand, budgets, and cost profiles into one prefix-complete target per worker. The planner maps measured profiles to $w_m(b)$ and either $f_m(b)$ or $f_{m\leftarrow h}(b)$, constructs exact node marginals, evaluates the guaranteed and refined candidates, and exports the best target set to the cache layer. Its state is limited to the prefix tree, demand arrays, worker budgets, and requester-side or pairwise cost tables.

When routing or popularity changes, the planner evaluates both the incumbent and a reoptimized target under refreshed demand. It adopts the new target only when the predicted horizon saving exceeds the one-time cost of materializing newly assigned chunks. Section~\ref{sec:eval} evaluates this amortization rule together with solver runtime.

\section{Performance Evaluation}
\label{sec:eval}

The evaluation follows the paper's claim chain through five Research Questions (RQs): (RQ1) do measured costs require worker-specific decisions; (RQ2) how close is PrefixPlace to exact and relaxed optima; (RQ3) when does coordination create value; (RQ4) do the structural predictions persist across workloads, source-dependent costs, and input perturbations; and (RQ5) can PrefixPlace replan efficiently under demand shifts?

\subsection{Experimental methodology}
\paragraph{Measured compute and KV profiles} We profile unquantized FP16 Qwen2.5 models with identical prompts and software settings. The cross-hardware study runs 3B on T4, L4, and A100; 7B on L4 and A100; and 14B and 32B on A100. Each context length uses five prompts, three warm-ups, and seven timed repetitions; incremental measurements use three prompts and five repetitions per prefix--segment pair. Compute Unified Device Architecture (CUDA) synchronization brackets every timed region, and all repetitions are retained. The stack is CUDA 12.8, PyTorch 2.8.0+cu128, and Transformers 4.55.2. Table~\ref{tab:profile} summarizes the completed context ranges, measured KV footprints, and repetition stability. Unless otherwise stated, the main experiments use 512-token placement chunks and the corresponding measured incremental curves, with linear interpolation between sampled depths. Effective KV goodputs from 0.75 to 5 Gb/s span the measured fetch/recompute boundaries; RQ4 evaluates 1024-token chunks as a granularity sensitivity check.

\paragraph{Parameterized and public prefix structures} Two reproducible families vary sharing, locality, skew, and depth independently. The RAG-shaped workload has a shared system prefix, 40 document branches of 4--6 chunks, and 10 query leaves per document, with Zipf(0.9) document and Zipf(0.6) query popularity. A shared-request fraction $\rho$ distributes that fraction of each document's requests uniformly across workers; the remainder stays at a random home worker. The multi-turn workload has 120 session chains of 3--15 chunks with Zipf(0.8) popularity and the same attachment model. Unless varied, each worker can store 10\% of tree nodes.

We also derive a public retrieval tree from all English WikiQA~\cite{wikiqa} entries (29,258 candidate rows, 3,047 questions, 2,811 document titles), which a deterministic lexical-token estimator partitions into 9,917 128-unit placement chunks; under the measured 3B FP16 KV footprint, an equal 1-GiB budget holds 227 chunks per worker. Offline placements use the exact aggregate demand of the full corpus. To evaluate the order-sensitive vLLM Automatic Prefix Caching (vLLM-APC) baseline under that same demand, each trial forms a complete 18,282-request multiset containing every worker--question pair exactly once; independent random permutations warm and evaluate the cache. A $\pm25\%$ lexical-token-scale sweep tests sensitivity to the token mapping.

\paragraph{Unified benchmark with exact optima} We evaluate 432 three-worker instances and solve every instance to exact Mixed-Integer Linear Programming (MILP) optimality. The first 216-instance factorial block uses a reduced RAG-shaped tree designed for exact solution (30 document branches of two to four chunks, three question leaves each; 175--189 chunks), crossing eight workload seeds, three request-sharing levels, three transfer-goodput settings, and three worker configurations (mixed T4/L4/A100, three L4, or three A100; $8\times3\times3\times3=216$). The second block (22 branches of two to nine chunks, three leaves each; 183 chunks; mixed T4/L4/A100) crosses three levels of popularity-to-depth mismatch, two demand-skew levels, two budgets, three sharing levels, two goodput settings, and three seeds ($3\times2\times2\times3\times2\times3=216$). At stronger mismatch levels, popular documents are assigned more often to shallower branches, making demand alone less predictive of the materialization cost saved by placement. All exact quality statistics below are computed over the pooled 432-instance benchmark.

\paragraph{Comparators and reporting metric} We compare PrefixPlace with four baselines. \emph{vLLM-APC}~\cite{vllm,vllm-apc-docs} replays vLLM's per-worker block-level Least Recently Used (LRU) prefix cache on the routed request stream. \emph{Independent} runs the exact rooted-tree oracle per worker on local demand weighted by recomputation cost, ignoring cross-worker reuse value; \emph{Popularity} uses local demand alone; \emph{LocalDedup} additionally discounts chunks already covered by earlier workers but not the requester-wide value created by the first copy. The requester-side placement comparator is the largest-$F$ result among Independent, Popularity, and LocalDedup; source-dependent experiments use the strongest of mean-rate, mean-transfer-cost, best-source-cost, and Independent variants. For the RAG request-stream replay, a 30,000-request epoch warms APC and profiles the planners, and an independent 30,000-request epoch evaluates all methods across five workload $\times$ five arrival-order seeds. For WikiQA, the offline methods use exact aggregate demand, while vLLM-APC uses 25 independent warm/evaluation permutation pairs of the complete worker--question request multiset. Saving over all-recompute is the profiled recomputation cost avoided by local hits or planned local/remote reuse. Because all cost inputs are measured on the profiled GPUs, the reported objective is denominated in milliseconds of avoided recomputation and transfer per epoch, and gains translate directly to materialization-time reductions on the corresponding hardware. For PrefixPlace placement $P$ and comparator $B$, $\operatorname{gain}(P,B)=100[G(P)-G(B)]/G(B)$, where $G$ is $F$ or $F_{\rm pair}$. Comparisons against exact MILP solutions and tree-aware LP bounds separately assess absolute solution quality.

\paragraph{Exact solvers and demand transitions} The unified requester-side benchmark uses a MILP with binary placement $x_{m,b}$ and coverage $y_b$; workload-scale checkpoints use its LP relaxation:
\begin{align}
\max\;&\sum_{m,b}a_m(b)x_{m,b}+\sum_b\Delta(b)y_b \label{eq:lp}\\
\text{s.t. }&\sum_bx_{m,b}\le k_m,\quad x_{m,b}\le x_{m,\mathrm{par}(b)},\nonumber\\
&y_b\le\sum_mx_{m,b},\quad 0\le x_{m,b},y_b\le1.\nonumber
\end{align}
SciPy 1.17.0 and HiGHS 1.8.0 solve both formulations. The source-dependent MILP adds assignment variables $z_{m,h,b}$ with $z_{m,h,b}\le x_{h,b}$ and $\sum_hz_{m,h,b}\le1$, exactly linearizing~\eqref{eq:pair}. For replanning, 100 transitions span both parameterized workloads, five seeds, 1.0/1.5-Gb/s goodput, six mixed-GPU workers, and 10\% budgets. We shift 25/50/100\% of routing mass or perturb popularity with log-normal $\sigma=0.5/1.0$. Break-even requests equal the profiled one-time cost of newly assigned chunks divided by per-request objective gain.

\begin{table}[t]
\centering
\caption{Measured GPU--model profiles. The median Coefficient of Variation (CV) summarizes repeated full-prefill measurements.}
\label{tab:profile}
\footnotesize
\setlength{\tabcolsep}{5.0pt}
\begin{tabular}{@{}lrrr@{}}
\toprule
Profile & Context (tokens) & KV bytes/token & Median CV\\
\midrule
T4-3B & 128--7168 & 36,864 & 0.51\%\\
L4-3B & 128--8192 & 36,864 & 0.42\%\\
L4-7B & 128--8192 & 57,344 & 0.35\%\\
A100-3B & 128--8192 & 36,864 & 0.27\%\\
A100-7B & 128--8192 & 57,344 & 0.25\%\\
A100-14B & 128--8192 & 196,608 & 0.17\%\\
A100-32B & 128--4096 & 262,144 & 0.16\%\\
\bottomrule
\end{tabular}
\end{table}

\subsection{RQ1: do measured profiles require worker-specific decisions?}
Table~\ref{tab:profile} first establishes the coverage and stability of the measured profiles; median CV is at most 0.51\% for every GPU--model pair. Figure~\ref{fig:hw} then holds model, representation, prompts, and effective KV goodput fixed and shows the decision consequence: T4, L4, and A100 fall into fetch-dominated, crossover, and recompute-dominated regimes. For Qwen2.5-7B, the measured 512-token recomputation curve is approximately $170.2+0.0100p$ ms on L4 and $48.1+0.0050p$ ms on A100, where $p$ is existing-prefix length. At 3 Gb/s, the same 29.4-MB chunk favors fetching on L4 but crosses near 6.1K tokens on A100. For A100-32B, the 134.2-MB chunk crosses near 2.7K tokens at 5 Gb/s. Thus, hardware identity changes the preferred miss action even when the model and network condition are held fixed.

\subsection{RQ2: how close is PrefixPlace to optimum?}
\paragraph{Exact optimum} Figure~\ref{fig:exact} scores vLLM-APC and the principal offline methods under the exact MILP objective on the unified 432-instance benchmark, with 175--189 chunks per instance. Per instance, five independent 30,000-request streams warm vLLM-APC; each frozen cache state, restricted to its usable prefix-complete blocks, is scored under the same objective. Averaged over the five orders, vLLM-APC reaches 78.63\% of optimum with a 65.62\% minimum. PrefixPlace averages 99.84\% and never falls below 98.02\%. These results demonstrate that PrefixPlace remains consistently near-optimal across the unified benchmark.

\begin{figure}[t]
\centering
\includegraphics[width=.80\linewidth]{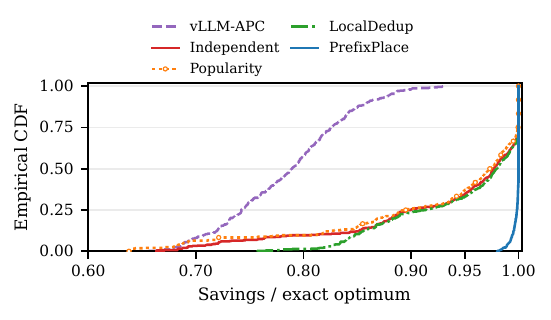}
\caption{Empirical Cumulative Distribution Function (CDF) of objective ratios to the exact MILP optimum across the unified 432-instance benchmark, with 175--189 chunks per instance. vLLM-APC freezes each replayed cache state and is evaluated under the same objective.}
\label{fig:exact}
\end{figure}

\paragraph{Component ablation} Table~\ref{tab:ablation} isolates how complete PrefixPlace reaches this quality: WorkerGreedy averages 98.15\% of optimum; coordinate refinement improves 124 instances (mean 98.93\%, minimum 82.33\% to 85.76\%); order-diverse starts improve another 161 (mean 99.84\%, minimum 97.55\%); and the comparator safeguards lift the minimum to 98.02\%. The global minimum at every stage lies in the popularity-to-depth-mismatch block, the difficult tail of the benchmark, which refinement, order diversity, and safeguards repair while keeping every instance within 1.98\% of the exact optimum.

\begin{table}[t]
\centering
\caption{Cumulative component ablation on the unified 432-instance exact benchmark.}
\label{tab:ablation}
\footnotesize
\setlength{\tabcolsep}{3.5pt}
\begin{tabular}{@{}lrrr@{}}
\toprule
Variant & Mean/Opt. & Min/Opt. & Exact\\
\midrule
One-pass WorkerGreedy & 98.15\% & 82.33\% & 174/432\\
+ coordinate refinement & 98.93\% & 85.76\% & 196/432\\
+ order-diverse starts & 99.84\% & 97.55\% & 246/432\\
+ comparator safeguards & 99.84\% & 98.02\% & 246/432\\
\bottomrule
\end{tabular}
\end{table}

\paragraph{Broader exact and relaxed benchmarks} Table~\ref{tab:quality} consolidates this result with source-dependent exact instances and workload-scale LP bounds. Across 45 source-dependent MILP instances, PrefixPlace averages 99.62\% of optimum and never falls below 97.90\%. The 40 tree-aware LP checkpoints (three T4, three L4, three A100, or mixed T4/L4/A100 workers; two sharing levels; five seeds; $4\times2\times5=40$) yield 99.66\% of the LP upper bound on average, never below 98.71\%, with the 10-checkpoint mixed-GPU subset at 99.55\% and never below 99.07\%. Because the LP relaxation retains capacity and parent closure and upper-bounds the integral optimum, each LP ratio is a valid lower bound on PrefixPlace's ratio to it.

\begin{table}[t]
\centering
\caption{Solution quality against exact optima and LP upper bounds.}
\label{tab:quality}
\scriptsize
\setlength{\tabcolsep}{3.2pt}
\begin{tabular}{@{}lrr@{}}
\toprule
Setting & Mean ratio & Minimum ratio\\
\midrule
Requester-side exact (432) & 99.84\% & 98.02\%\\
Source-dependent exact (45) & 99.62\% & 97.90\%\\
Tree-aware LP (40) & 99.66\% & 98.71\%\\
Mixed-GPU LP subset (10) & 99.55\% & 99.07\%\\
\bottomrule
\end{tabular}
\end{table}

\subsection{RQ3: when does coordination create value?}
Figure~\ref{fig:phase}(a) fixes full request sharing and the default 10\% worker budget, then sweeps effective KV goodput in the mixed T4/L4/A100 RAG-shaped workload. Every point is a five-seed mean with a 95\% confidence interval. PrefixPlace improves over the strongest evaluated placement baseline throughout the measured range: gain rises from 1.1\% at 0.75 Gb/s to 6.2\% at 1.5 Gb/s and remains 8.1--8.9\% from 2 to 5 Gb/s. The peak mean is 8.9\% at 3 Gb/s, and the maximum over the complete 600-setting RAG sweep is 10.0\%.

\begin{figure}[t]
\centering
\includegraphics[width=.98\linewidth]{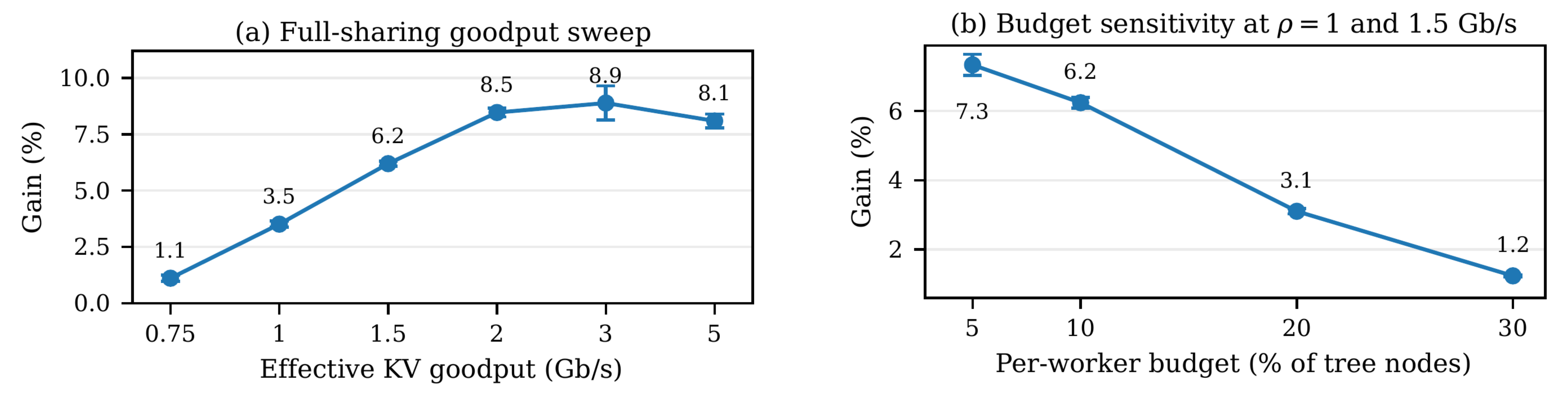}
\caption{Value unlocked by coordination over the strongest evaluated placement baseline. (a) Full-sharing effective-KV-goodput sweep at a 10\% worker budget. (b) Gain as per-worker budget varies under full sharing and 1.5-Gb/s effective KV goodput. Error bars show 95\% confidence intervals; numeric labels report mean gain.}
\label{fig:phase}
\end{figure}

Figure~\ref{fig:phase}(b) isolates capacity. Mean gains are 7.34\%, 6.24\%, 3.10\%, and 1.24\% at 5\%, 10\%, 20\%, and 30\% budgets, respectively. Scarce capacity makes redundant replicas more expensive, so coordination helps most precisely where placement choices are consequential; the candidate safeguard retains the strongest evaluated alternative under the same objective.

\subsection{RQ4: do the structural predictions persist across workloads and costs?}
\paragraph{RAG request replay} Figure~\ref{fig:public}(a) reports PrefixPlace's paired relative gain over each comparator on the RAG request replay. Across 25 paired runs, PrefixPlace improves saving by 40.29\% over vLLM-APC, 8.59\% over both Independent and Popularity (unrounded 8.588\% and 8.595\%, reflecting their nearly identical saving), and 6.30\% over LocalDedup, and every paired run favors PrefixPlace.

\begin{figure}[t]
\centering
\includegraphics[width=.99\linewidth]{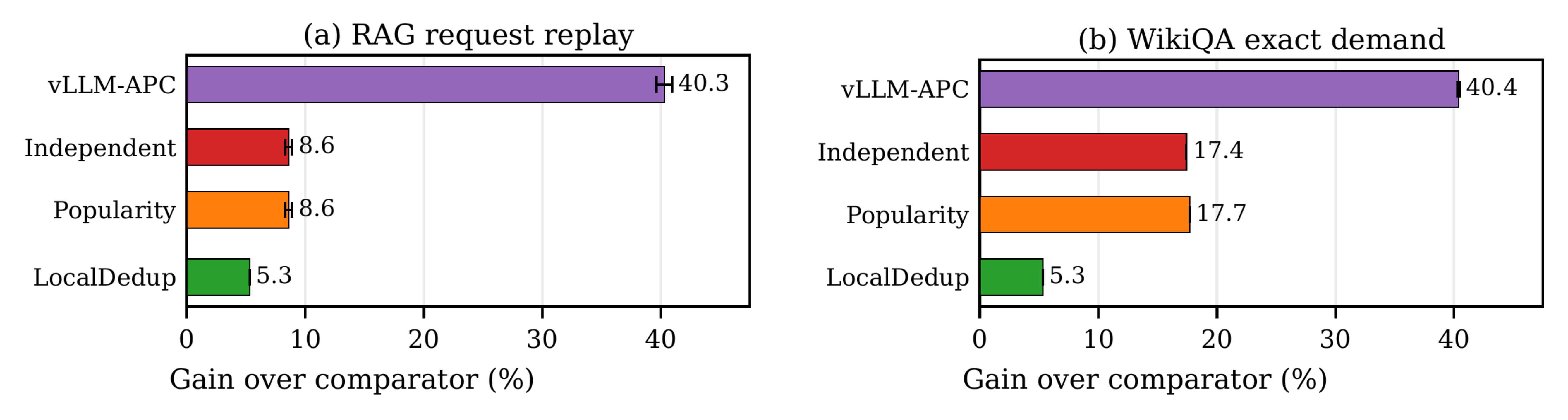}
\caption{PrefixPlace's relative gain over each comparator. (a) RAG request replay: means and 95\% confidence intervals over 25 paired 30,000-profile/30,000-evaluation runs. (b) WikiQA exact-demand evaluation: offline bars use the full aggregate demand; the APC bar averages 25 independent warm/evaluation permutations of the complete worker--question request multiset, preserving that same demand.}
\label{fig:public}
\end{figure}

\paragraph{WikiQA exact-demand evaluation} Figure~\ref{fig:public}(b) uses the exact aggregate demand of the full WikiQA workload. PrefixPlace improves saving by 40.41\% over vLLM-APC, 17.43\% over Independent, 17.72\% over Popularity, and 5.31\% over LocalDedup. For vLLM-APC, the 40.41\% mean is computed over 25 independent order pairs and has a 95\% confidence interval of $\pm0.10$ percentage points; each replay preserves the same complete worker--question demand used by the offline methods. The gain over the strongest offline baseline remains 4.99--6.69\% under a $\pm25\%$ lexical-token-scale sweep.

\paragraph{Source-dependent transfer costs} Figure~\ref{fig:topology} evaluates whether holder identity changes placement value at shared-request fractions $\rho\in\{0.75,1.0\}$. Two groups each contain one T4, L4, and A100 worker, with 10-Gb/s within-group goodput and 1--3-Gb/s between-group goodput; source-aware PrefixPlace uses the full $f_{m\leftarrow h}$ table against the strongest source-oblivious variant under true costs. At 1 Gb/s between groups, mean gain is 5.8\% for $\rho=0.75$ and 7.4\% for $\rho=1.0$ (maximum 8.1\% over all 30 settings); at 3 Gb/s, gains remain 2.0\% and 1.8\%, so source identity matters beyond requester-only summaries.

\begin{figure}[t]
\centering
\includegraphics[width=.80\linewidth]{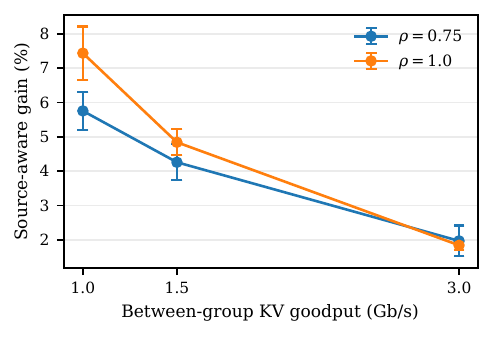}
\caption{Source-aware gain over the strongest evaluated source-oblivious variant at $\rho=0.75$ and $\rho=1.0$ as between-group effective KV goodput varies (mean $\pm1.96$ standard errors across five seeds).}
\label{fig:topology}
\end{figure}

\setcounter{table}{6}
\begin{table*}[!t]
\centering
\caption{Optimization scope of representative KV-reuse methods.}
\label{tab:scope}
\fontsize{6.0pt}{6.7pt}\selectfont
\setlength{\tabcolsep}{3.2pt}
\renewcommand{\arraystretch}{0.86}
\begin{tabularx}{\textwidth}{@{}>{\raggedright\arraybackslash}p{0.11\textwidth}>{\raggedright\arraybackslash}p{0.20\textwidth}>{\raggedright\arraybackslash}p{0.22\textwidth}>{\raggedright\arraybackslash}X@{}}
\toprule
Work & Primary decision & KV object / scope & Distinction from PrefixPlace \\
\midrule
vLLM~\cite{vllm}/\newline SGLang~\cite{sglang} & Local cache organization and execution & Exact-prefix blocks / radix tree & Provide worker-local exact-prefix reuse; vLLM-APC is our replayed baseline. PrefixPlace plans coordinated cross-worker targets. \\
\addlinespace[1.5pt]
Preble~\cite{preble} & Request-to-worker assignment & Prefix locality in existing worker caches & Optimizes routing; PrefixPlace plans prefix-complete placement from an epoch demand matrix. \\
\addlinespace[1.5pt]
Mooncake~\cite{mooncake}/\newline IMPRESS~\cite{impress2025}/\newline CacheGen~\cite{cachegen} & Transfer, storage, and tier management & Streamed, disaggregated, or multi-tier prefix KV & Provide transfer mechanisms and costs; PrefixPlace selects coordinated placement targets. \\
\addlinespace[1.5pt]
RAGCache~\cite{ragcache}/\newline HotPrefix~\cite{hotprefix}/\newline UniCache~\cite{unicache} & Hierarchy scheduling or eviction & RAG knowledge trees or worker-local prefix caches & Optimize eviction or hierarchy decisions rather than coordinated prefix-complete placement. \\
\addlinespace[1.5pt]
CacheBlend~\cite{cacheblend} & Request-local fused KV state & Non-prefix retrieved RAG chunks & Recomputes within a request rather than coordinating exact-prefix placement. \\
\addlinespace[1.5pt]
SemCache~\cite{semcache2026} & Semantic-aware cache sharing & Multi-user inference with Low-Rank Adaptation (LoRA) at the edge & Coordinates semantic sharing rather than prefix-complete exact-prefix placement. \\
\addlinespace[1.5pt]
\textbf{PrefixPlace} & \textbf{Prefix-complete target per worker} & \textbf{Exact-prefix placement chunks} & \textbf{Optimizes coordinated placement under profiled compute and transfer costs.} \\
\bottomrule
\end{tabularx}
\end{table*}
\setcounter{table}{3}

\paragraph{Robustness analysis} Table~\ref{tab:robustness} evaluates placements computed from perturbed inputs under the unperturbed calibrated objective, perturbing each profiled or demand value independently as $\widehat c=c\exp(\epsilon)$, $\epsilon\sim\mathcal N(0,\sigma^2)$. Even at $\sigma=0.30$, the fifth percentile retains at least 98.74\% of calibrated-placement saving.

\begin{table}[t]
\centering
\caption{Robustness analysis under profile and demand perturbations.}
\label{tab:robustness}
\scriptsize
\setlength{\tabcolsep}{3.8pt}
\begin{tabular}{@{}lrr@{}}
\toprule
Perturbation & Mean retained & Fifth percentile\\
\midrule
Profile noise $\sigma=0.20$ & 99.70\% & 99.23\%\\
Profile noise $\sigma=0.30$ & 99.34\% & 98.75\%\\
Demand noise $\sigma=0.30$ & 99.36\% & 98.74\%\\
\bottomrule
\end{tabular}
\end{table}

\paragraph{Sensitivity analysis} Table~\ref{tab:sensitivity} varies workload structure, demand skew, worker order, placement granularity, and request-stream horizon: the multi-turn workload retains measurable coordination value, increasing Zipf concentration reduces new coverage opportunities, coordinate refinement nearly removes order sensitivity, doubling the chunk size preserves the coordination pattern, and the advantage over vLLM-APC persists across replay horizons of 3,000 to 30,000 requests.

\begin{table}[t]
\centering
\caption{Sensitivity analysis across workload and algorithm settings.}
\label{tab:sensitivity}
\scriptsize
\setlength{\tabcolsep}{2.6pt}
\begin{tabular}{@{}p{0.25\linewidth}p{0.68\linewidth}@{}}
\toprule
Variation & Result\\
\midrule
Multi-turn structure & 45 settings: 1.83\% mean and 4.93\% maximum gain\\
Zipf exponent $0.5\!\rightarrow\!1.5$ & Mean gain decreases from 13.53\% to 2.27\%\\
100 worker orders & Refinement adds 4.12\%; CV falls from 0.98\% to 0.09\%\\
1024-token chunks & 27 settings: 3.00\% mean and 11.01\% maximum gain\\
RAG replay horizon, 3,000--30,000 requests & vs. APC: 38.91--40.54\%; vs. best offline: 4.50--6.26\%\\
\bottomrule
\end{tabular}
\end{table}

\begin{table}[t]
\centering
\caption{Reoptimization under refreshed demand. Gain and BE statistics are computed over beneficial updates; P90 denotes the 90th percentile.}
\label{tab:epoch-update}
\scriptsize
\setlength{\tabcolsep}{2.5pt}
\begin{tabular}{@{}lrrrr@{}}
\toprule
Demand change & Improves & Gain & Median BE & P90 BE\\
\midrule
Popularity $\sigma=0.5$ & 20/20 & 0.87\% & 1,202 & 2,442\\
Popularity $\sigma=1.0$ & 20/20 & 3.84\% & 482 & 854\\
Routing 25\% & 15/20 & 0.61\% & 1,637 & 7,473\\
Routing 50\% & 20/20 & 3.46\% & 909 & 1,113\\
Routing 100\% & 20/20 & 23.74\% & 251 & 390\\
\bottomrule
\end{tabular}
\end{table}

\subsection{RQ5: can PrefixPlace replan efficiently under demand shifts?}
\paragraph{Placement updates} Table~\ref{tab:epoch-update} reports reoptimization gains and Break-Even (BE) requests. Reoptimization improves 95 of 100 transitions, with the amortization rule retaining the incumbent in five mild 25\% routing shifts; across beneficial updates, median and 90th-percentile BE are 769 and 2,346 requests, 96.8\% break even within 5,000 and all within 10,000, and complete routing shifts yield 23.74\% median gain with a 251-request median BE.

\paragraph{Solver runtime} Figure~\ref{fig:scale} measures CPU time of the coordinate-refined solver as tree size and worker count increase, with per-worker budgets given in the caption. On an EPYC 9V74 CPU, the $n=10{,}000$, $k=128$ configuration takes 0.41, 1.12, and 2.35 s for 4, 8, and 16 workers, and the largest configuration ($n=50{,}000$, $M=16$, $k=128$) completes in 12.29 s; refinement converges in two to three sweeps.

\begin{figure}[t]
\centering
\includegraphics[width=.80\linewidth]{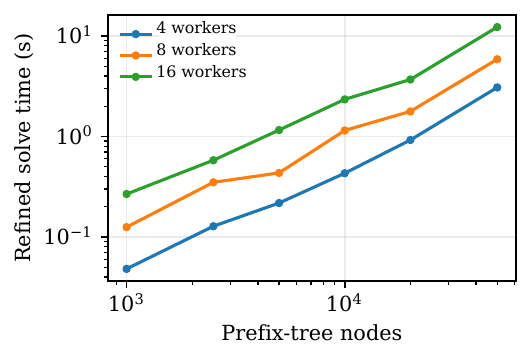}
\caption{CPU time of the coordinate-refined solver (log scale) as tree size $n$ and worker count $M$ increase. Each worker receives $k=\min(128,\max(16,\lfloor n/20\rfloor))$ chunks; the $n=50{,}000$, $M=16$, $k=128$ case completes in 12.29 s.}
\label{fig:scale}
\end{figure}

\section{Related Work}
\label{sec:related}

Table~\ref{tab:scope} positions PrefixPlace against representative KV-reuse decisions, including local reuse, transfer and tiering, hierarchy or eviction, and request routing.

\noindent\textbf{Prefix reuse, transfer, and eviction.}
vLLM~\cite{vllm} provides hash-based Automatic Prefix Caching with block-level LRU eviction~\cite{vllm-apc-docs}, SGLang~\cite{sglang} radix-tree exact-prefix reuse, and PromptCache~\cite{promptcache} modular attention reuse. CachedAttention~\cite{cachedattention}, CacheGen~\cite{cachegen}, Mooncake~\cite{mooncake}, and IMPRESS~\cite{impress2025} optimize KV movement across device, network, or storage paths; RAGCache~\cite{ragcache}, HotPrefix~\cite{hotprefix}, UniCache~\cite{unicache}, and workload characterization~\cite{kvcachewild2025} study hierarchy or eviction; CacheBlend~\cite{cacheblend} reuses non-prefix RAG chunks, DroidSpeak~\cite{droidspeak} shares state across model variants, and SemCache~\cite{semcache2026} coordinates semantic sharing for Low-Rank Adaptation (LoRA)-based edge inference.

\noindent\textbf{Routing and cooperative placement.}
Preble~\cite{preble} routes requests toward existing prefix locality; PrefixPlace treats the routing-induced demand matrix as input and plans prefix-complete placement. Classical paging~\cite{sleator-tarjan} and Landlord-style file caching~\cite{landlord} assume placement-independent object costs, while FemtoCaching~\cite{femtocaching} captures cooperative placement. The present problem couples a rooted-tree feasible set at each worker with requester-specific recomputation and transfer value.

\noindent\textbf{Optimization tools.}
Tree knapsack~\cite{johnson-niemi} and submodular maximization~\cite{nwf78} provide related techniques. Maximum separable assignment~\cite{fgms06} supplies stronger configuration-LP guarantees for the requester-side separable special case. PrefixPlace instead develops one direct rooted-tree-oracle algorithm for both requester-side coverage and source-dependent facility-location marginals; monotone refinement then produces near-optimal solutions at the evaluated scales.

\section{Conclusions and Future Work}
\label{sec:conclusion}
Reusable prefix KV states create a cost-aware placement problem beyond hit-rate policies, which PrefixPlace addresses with modular-plus-coverage and facility-location formulations, an exact $O(nk)$ rooted-tree oracle, a coordinated $1/2$-approximation, and monotone refinement. Across 432 requester-side and 45 source-dependent instances solved to exact MILP optimality, it attains 99.84\% and 99.62\% of optimum on average, never below 98.02\% and 97.90\%, respectively, and improves materialization-cost saving by 40.3\% over vLLM-APC and 6.3\% over the best offline baseline on RAG replays, with consistent gains on WikiQA, source-aware topologies, and demand shifts. A 50,000-node, 16-worker placement solves in 12.3 s on one CPU, and beneficial updates break even within a median of 769 requests. Future work includes co-optimizing request routing with placement and extending the epoch-level model to online placement with regret guarantees.

\bibliographystyle{IEEEtran}
\bibliography{Bibliography}

\end{document}